%% file: DDP_jsp.tex
\newcommand{\BE}{\begin{eqnarray}}
\newcommand{\EE}{\end{eqnarray}}
\newcommand{\be}{\begin{eqnarray}}
\newcommand{\ee}{\end{eqnarray}}
\newcommand{\BEN}{\begin{eqnarray*}}
\newcommand{\EEN}{\end{eqnarray*}}
\newcommand{\ben}{\begin{eqnarray*}}
\newcommand{\een}{\end{eqnarray*}}
\newcommand{\BA}{\begin{array}}
\newcommand{\EA}{\end{array}}
\newcommand{\ba}{\begin{array}}
\newcommand{\ea}{\end{array}}
\newcommand{\BLN}{\begin{align*}}
\newcommand{\ELN}{\end{align*}}

\newcommand{\pt}{\partial}
\newcommand{\cd}{\cdot}
\newcommand{\f}{\frac}

\newcommand{\bl}{\begin{lstlisting}}
\newcommand{\el}{\end{lstlisting}}

\newcommand\disp{\displaystyle}

\newcommand{\Li}{\operatorname{Li}}

\newcommand{\op}[1]{\operatorname{#1}}

\documentclass[12pt,a4paper,twoside]{article}

\usepackage{enumerate}
\usepackage[english]{babel}
\usepackage[latin1]{inputenc}
\usepackage{subfigure}
\usepackage[T1]{fontenc}
\usepackage{graphicx}
\usepackage{latexsym}
\usepackage{amsmath,amssymb,mathrsfs}
\usepackage{ifpdf}
\usepackage{cancel} 
\usepackage{amsmath, amsfonts, amsthm, amssymb, amscd}
\usepackage{graphicx}
\usepackage{makeidx}
\usepackage{sectsty}
\usepackage{color}
\usepackage{relsize} 

\usepackage[format=hang,margin=0.1cm,textfont=footnotesize,labelfont=up]{caption}[2013/02/03] 
\usepackage{eufrak}
\usepackage{fancyhdr}
\usepackage{hyperref}

\usepackage{tikz}
\usetikzlibrary{shapes.arrows}

\sectionfont{\large}
\subsectionfont{\normalsize}

\newtheoremstyle{theorems}  
  {\topsep}   
  {\topsep}   
  {\itshape}  
  {10pt}      
  {\bfseries} 
  {.}         
  {5pt} 	  
  {}          
  
\newtheoremstyle{definitions}  
  {\topsep}   
  {\topsep}   
  {}          
  {10pt}      
  {\bfseries} 
  {.}         
  {5pt} 	  
  {}          

\theoremstyle{theorems}

\newtheorem{proposition}{Proposition}[section]
\newtheorem{theorem}{Theorem}[section]
\newtheorem{lemma}{Lemma}[section]

\newtheorem{definition}{Definition}[section]

\begin{document}

\pagestyle{myheadings}
\fancyhf{}
\fancyhead[L]{}
\fancyhead[R]{}

\title{\Large{Higher-order Airy scaling in deformed Dyck paths}}

\author{\normalsize{Nils Haug\footnote{School of Mathematical Sciences, Queen Mary University of London, London E1 4NS, UK}, Adri Olde Daalhuis\footnote{School of Mathematics, University of Edinburgh, Edinburgh EH9 3JZ, UK}~~and Thomas Prellberg$^*$}}

\date{\normalsize{\today}}

\maketitle

\begin{abstract}
\footnotesize{\noindent We introduce a deformed version of Dyck paths (DDP), where additional to the steps allowed for Dyck paths, `jumps' orthogonal to the preferred direction of the path are permitted. We consider the generating function of DDP, weighted with respect to their half-length, area and number of jumps. This rep\-re\-sents the first example of an exactly solvable two-dimensional lattice vesicle model showing a higher-order multi\-critical point. Applying the generalized method of steepest descents, we see that the associated two-variable scaling function is given by the logarithmic derivative of a generalized (higher-order) Airy integral.}
\end{abstract}

\sectionfont{\normalsize}
\subsectionfont{\normalsize}

\section{Introduction}
\label{section:introduction}

Biological vesicles act as containers transporting molecules inside cells, and consist of a lipid membrane enclosing a fluid \cite{Alberts07}. Depending on the pressure difference between the inside and the outside, the vesicles can be found in a deflated or an inflated phase. When analysing the transition between these two phases using statistical mechanics, the pressure difference is modelled by introducing a volume fugacity into the partition function. In a two-dimensional setting, the vesicles can be described by different subclasses of two-dimensional self-avoiding lattice polygons (SAP) \cite{Leibler89,Fisher91}, in the same way as self-avoiding walks are used to model the behaviour of long polymer chains subject to volume interactions \cite{Flory53}. Mathematically, the corresponding grand-canonical partition function is the generating function of the polygons, weighted with respect to their perimeter and their area. 

On the basis of exact enumeration data, it was conjectured in \cite{Richard01} that for vesicles modelled by unrestricted, rooted SAP, the pressure induced phase transition is characterized by a tri-critical point, with the associated single-variable scaling function being given by the logarithmic derivative of the Airy function, which for $s\in\mathbb{C}$ is defined as
\be 
\op{Ai}(s) = \f{1}{2\pi i}\int_{e^{-i\pi/3}\infty}^{e^{i\pi/3}\infty} \exp\left(\f{u^3}{3}-s u\right)du
\label{eq:def_Airy_function}
\ee
(\href{http://dlmf.nist.gov/9}{chap. 9} in \cite{NIST}). By applying the method of steepest descents to the exact solution for the generating function, the same type of scaling behaviour has been proven rigorously to hold for staircase polygons and Dyck and Schr\"oder paths \cite{Prellberg95,Haug15,Haug16}. These models can be interpreted as directed subclasses of SAP, with a lower sawtooth boundary in the latter two cases. The fact that their scaling function is the same as the one conjectured for unrestricted, rooted SAP supports the assumption that close to the phase transition, the overhangs present in unrestricted SAP can be neglected and therefore the restriction to directed polygons does not change the qualitative behaviour of the model. Further evidence for this comes from the fact that the asymptotic behaviour of the mean perimeter in the infinite area ensemble is the same both for convex and column-convex polygons \cite{Mitra10}. The difference between these two models is that in the latter, overhangs in one direction of the lattice are allowed, which makes it plausible that the asymptotic behaviour of the model is also not changed when overhangs in the second direction are allowed.

Neglecting overhangs is also usual to consider in interface physics, for example by using a solid-on-solid model to describe the boundary of oppositely magnetized domains in the Ising model or the shape of liquid drops on a substrate at low temperatures \cite{Temperley52,Dietrich88}.

In \cite{Cardy01}, it was postulated that for SAP, there exists an entire hierarchy of higher-order scaling functions given for $k\geq 3$ by the loga\-rithmic deri\-vative $(\pt/\pt s_1)\Theta_k(s_1,\dots,s_{k-2})$, where for $s_1,s_2,\dots,s_{k-2}\in\mathbb{C}$, and
\be 
\Theta_k(s_1,\dots,s_{k-2}) = \f{1}{2\pi i}\int_{e^{-i\pi/k}\infty}^{e^{i\pi/k}\infty} \exp\left(\f{u^k}{k}-\sum_{j=1}^{k-2}s_{j} u^{j}\right)du
\label{eq:generalized_Airy}
\ee
(\href{http://dlmf.nist.gov/36.2}{$§$36.2} in \cite{NIST}). This function can be seen as a generalized Airy function, since $\Theta_3(s)=\op{Ai}(s)$. We present here the first concrete example of a lattice polygon model with a higher-order multicritical point characterized by the two-variable scaling function 
\be 
\Phi(s_1,s_2) = \f{\pt}{\pt s_1}\ln\Big(\Theta_4(s_1,s_2)\Big)
\label{eq:Pearcey_scaling_function},
\ee
where $\Theta_4(s_1,s_2)$ is also called a Pearcey function.
It consists of a deformed version of Dyck paths (DDP), where additional to the steps $(1,0)$ and $(0,1)$ allowed for (standard) Dyck paths, `jump' steps in the direction $(-1,1)$ are allowed. For the reasons given above, we consider this model a valid simplification of two-dimensional vesicles for the purpose of studying their critical behaviour.

In Section \ref{section:model}, we will define DDP precisely and derive the functional equation for their generating function, weighted with respect to their area, their length and their number of jumps. An expression for the generating function in the form of a fraction of two basic hypergeometric series will be obtained in Section \ref{section:sol_funeq}. The main result is given in Section \ref{section:result}, and the remaining sections contain the steps of its derivation.  A contour integral representation for the series occurring in the generating function of DDP is derived in Section \ref{section:contintrep}, which in the limit $q\to 1^-$ has a leading contribution in the form of a saddle point integral. The location of the relevant saddle points depending on the parameters $w$ and $t$ is discussed in Section \ref{section:location_of_the_saddle_points}, and the geometry of the paths of steepest descent originating from them is investigated in Section \ref{section:geometry_of_paths_of_steepest_descent}. In Section \ref{section:canonicaltrafo}, the integral expression for the basic hypergeometric series is then transformed into a canonical form, and the asymptotic behaviour of the coefficients of this transformation around the multicritical point is analysed. The asymptotic expression for the basic hypergeometric series is then obtained by evaluating the transformed integral in Section \ref{section:asyofphi}, which directly leads to Theorem \ref{thm:scaling_around_multicritical_point}.

\section{The model}
\label{section:model}

The model of DDP is defined as follows.

\begin{definition}
For $m,s \in \mathbb{Z}_{\geq 0}$ and $s\geq 2m$, a deformed Dyck path (DDP) of half-length $m$ is a walk $(x_k,y_k)_{k=0}^{s}$ on $\mathbb{Z}^2$, such that $(x_0,y_0)=(0,0)$,  $(x_{s},y_{s})=(2m,2m)$ and $y_k\geq x_k$ for all $0 \leq k \leq s$. Moreover, if $(x_k,y_k)=(x,y)$ for $0\leq k < s$, then $(x_{k+1},y_{k+1})$ is either $(x,y+1)$ or $(x+1,y)$ or $(x-1,y+1)$, which we call an up-step, a down-step or a jump, respectively.
\end{definition}
We consider the generating function
\be 
G(w,t,q) = \sum_{k=0}^\infty\sum_{m=0}^\infty\sum_{n=0}^\infty p_{k,m,n}\,w^k\,t^m\,q^n,
\label{eq:def_gf_ddp}
\ee
where $p_{k,m,n}$ is the number of DDP with $k$ jumps, half-length $m$ and area $n$, with the area being defined as the number of full lattice cells enclosed between the path and the main diagonal $x=y$. Fig.\,\ref{fig:gen_dp} shows an example of a DDP with half-length 9, 3 jumps and area 12. 

\begin{figure}[htbp]
\centering
\includegraphics[width=0.6\textwidth]{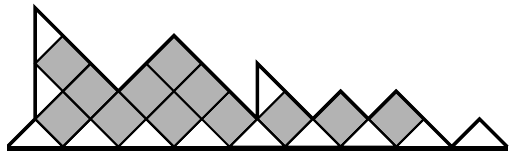}
\caption{A DDP of half-length 9, 3 jumps and area 12. The lattice is rotated such that its main diagonal lies horizontally in the image.}
  \label{fig:gen_dp}
\end{figure}

To obtain a functional equation for $G(w,t,q)$, we use the following factorization argument.  A DDP has either half-length zero, or it starts with an up-step followed by a DDP followed by a down-step and then another DDP, or it starts with a jump followed by a DDP followed by a down-step followed by another DDP followed by a down-step and then another DDP -- see Fig.\,\ref{fig:fun_eq_gen_dp} for an illustration.
\begin{figure}[h!]
  \centering
  \def\svgwidth{300pt}
  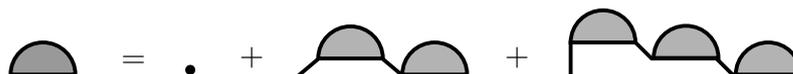
\caption{Graphical decomposition of the set of DDP, leading to Eq.\eqref{eq:fun_eq_gen_dp}.}
  \label{fig:fun_eq_gen_dp}
\end{figure}

From this decomposition we obtain
\be
w\,t\,G(q^2t)\,G(qt)\,G(t)+t\,G(qt)\,G(t)-G(t)+1&=&0,
\label{eq:fun_eq_gen_dp}
\ee
where $G(w,t,q)\equiv G(t)$ for brevity. Note that Eq.\eqref{eq:fun_eq_gen_dp} has a unique solution analytic at $t=0$. For $w=0$, it is satisfied by the generating function of Dyck paths, weighted with respect to their area and half-length \cite{Haug15}.

We also note that every DDP can be mapped uniquely onto a (standard) Dyck path by replacing every jump step by two consecutive up-steps. In this way, each Dyck path represents a family of DDP -- see Fig.\,\ref{fig:alternative_model}. The function $G(w,t,q)$ can therefore alternatively be interpreted as the generating function of Dyck paths, weighted with respect to their half-length and their area, with an additional weight $F_k(w/t,1/q)$ associated to each sequence of $k$ consecutive up-steps, followed by a down-step. Here, $F_k(s,q)$ is the generating function of appropriately weighted dimer coverings of an interval of length $k$ (for $q=1$, see \cite{Viennot89}).

\begin{figure}[h!]
\centering
\includegraphics[width=0.7\textwidth]{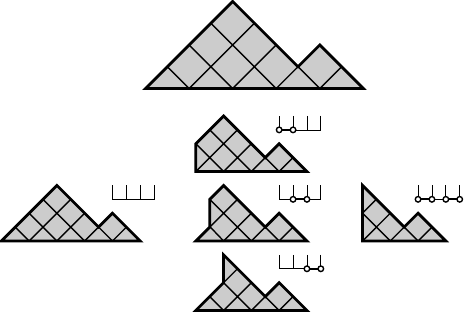}
\caption{A Dyck path (top) and the family of DDP representing it, together with the corresponding dimer coverings.}
  \label{fig:alternative_model}
\end{figure}

The $q$-Fibonacci polynomials $F_k(s,q)\equiv F_k$ satisfy the recurrence
\be
\label{fib_recurrence}
F_0=F_1=1\;,\quad F_n=F_{n-1}+q^{n-1} sF_{n-2}\quad\text{for $n\geq2$}
\ee
and are given explicitly by
\be
F_k(s,q)=\sum_{l=0}^{\lfloor k/2\rfloor}{k-l \brack l}_qq^{l^2}s^l\;,
\ee
see \cite{Carlitz75}. Using the dimer interpretation, a decomposition of Dyck paths by their left-most rise leads to an alternative functional equation for $G(w,t,q)\equiv G(t)$,
\be
G(t)=\sum_{k=0}^\infty t^kq^{\binom k2}F_k(w/t,1/q)\prod_{l=0}^{k-1}G(q^lt)\;.
\ee
For $q=1$, it follows from the inversion Lemma \cite{Viennot89} and the asymptotic behaviour of the number of dimer coverings \cite{Bousquet02} that $F_k(w/t,1)\geq 0$ for $w/t\geq -1/4$.

\section{Solution of the functional equation}
\label{section:sol_funeq}

Analogous to the case of Dyck and Schr\"oder paths \cite{Flajolet80,Haug16}, inserting the ansatz
\be 
G(w,t,q) = \f{H(w,qt,q)}{H(w,t,q)},
\label{eq:gf_generalized_DP}
\ee
into Eq.\eqref{eq:fun_eq_gen_dp} leads to the linearized functional equation
\be 
w t H(q^3 t)+ t H(q^2 t) - H(q t) + H(t) = 0,
\label{eq:linearized_fun_eq}
\ee
where $H(t) \equiv H(w,t,q)$. For $w,t,q \in \mathbb{C}$ and $|q|<1$, Eq.\eqref{eq:linearized_fun_eq} is solved by the basic hypergeometric series \cite{Gasper90}
\be 
H(w,t,q) = \phi(-w,t,q) =\sum_{n=0}^\infty \f{(-w;q)_n}{(q;q)_n}(-t)^n q^{n^2-n},
\label{eq:def_1_phi_2}
\ee
where for $n\in\mathbb{N}$ and $z,q\in\mathbb{C}$, $(z;q)_n = \prod_{k=0}^{n-1} (1-q^k z)$.

For $q\to 1^-$, both $\phi(-w,t,q)$ and $\phi(-w,q t,q)$ diverge and it is therefore not immediately clear which value $G(w,t,q)$ takes in this limit. But if we substitute $q=1$ into Eq.\eqref{eq:fun_eq_gen_dp}, then we obtain a cubic equation for $G(w,t,1)$, which is readily solved. In the special case $w=-1/9$, the radius of convergence of $G(w,t,1)$ is determined by a cubic root singularity at $t=1/3$ and around this value we therefore expect an area-length scaling behaviour which is qualitatively different from the Airy function scaling found for Dyck and Schr\"oder paths and staircase polygons. In order to analyse the asymptotics of the generating function in vicinity of the point $w=-1/9$, $t=1/3$ as $q\to 1^-$, we apply the method of steepest descent, generalized to the case of several coalescing saddle points.

\section{The main result}
\label{section:result}

The principal result of this paper is stated in the following theorem, which is an immediate consequence of Proposition \ref{prop:asymptotics_a<=1/9}.

\begin{theorem}
Let  $q=e^{-\epsilon}$, $\delta=\mathcal{O}\big(\epsilon^{1/2}\big)$ and $\tau=\f{3}{2}\delta+\mathcal{O}\big(\epsilon^{3/4}\big)$ as $\epsilon\to 0^+$. Then
\be
G\left(\delta-\f{1}{9},\f{1}{3}-\tau,q\right) = 3\,\Big(1 +2^{1/4}\,\Phi(s_1,s_2)\,\epsilon^{1/4}+\mathcal{O}\big(\epsilon^{1/2}\big)\Big),
\label{eq:scaling_relation}
\ee
as $\epsilon\to 0^+$, for all $s_1,s_2\in \mathbb{R}$ such that $|\Phi(s_1,s_2)|<\infty$, where
\begin{equation}
s_1=  3 \sqrt[4]{2}\left(\tau - \f{3}{2} \delta \right)\epsilon^{-3/4} ~\text{and }~  s_2=  \f{27\sqrt{2}}{8}\left(\delta + \f{1}{40} \tau^2 \right)\epsilon^{-1/2},
\label{eq:s1s2}
\end{equation}
and where $\Phi(s_1,s_2)$ is defined in Eq.\eqref{eq:Pearcey_scaling_function}.
\label{thm:scaling_around_multicritical_point}
\end{theorem}
For example, Theorem \ref{thm:scaling_around_multicritical_point} gives for all $s$ such that $|\Phi(s,0)|< \infty$,
\begin{equation}
G\left(-\f{1}{9},\f{1}{3}\left(1-s \epsilon^{\phi_{cr}}\right),q\right)=3 \left(1 +F\big(s\big)\epsilon^{\gamma_u}+\mathcal{O}\big(\epsilon^{1/2}\big)\right)
\label{eq:scaling_relation_delta=0}
\end{equation}
as $q=e^{-\epsilon}\to 1^-$, with $F(s)=\sqrt[4]{2}\Phi(\sqrt[4]{2} s,0)$, $\gamma_u=\f{1}{4}$ and $\phi_{cr}=\f{3}{4}$.

The exponents $\gamma_u$ and $\phi_{cr}$, together with $\gamma_t=\f{\gamma_u}{\phi_{cr}}=\f{1}{3}$ characterize the singular behaviour of $G(-\f{1}{9},t,q)$ around the multicritical point $(w,t,q)=(-\f{1}{9},\f{1}{3},1)$. The singular behaviour of $G\left(-\f{1}{9},\f{1}{3},1-\epsilon\right)$ as $\epsilon\to 0^+$ is determined by $\gamma_u$; $\gamma_t$ describes the singular behaviour of $G(-\f{1}{9},t,1)$ as $t\to \f{1}{3}^-$, and $\phi_{cr}$ is called the crossover exponent of the model.

The multcritical point for $w=-\f{1}{9}$ is the endpoint of a line of tri-critical points $(w,t,q)=(w,t_c(w),1)$ for $w>-\f{1}{9}$, which are characterized by the exponents $\gamma_u=\f{1}{3}$, $\gamma_t=\f{1}{2}$ and hence $\phi_{cr}=\f{2}{3}$ (see \cite{Rensburg00} for a general introduction to tri-critical scaling). The special case $w=0$ was analysed in \cite{Haug15}. In Table \ref{table_critexp}, we summarize the values of the critical exponents for $w=-\f{1}{9}$ and $w>-\f{1}{9}$.

\begin{table}
\begin{center}
\begin{tabular}{rccc}
		$w~~$		& ~~$\gamma_t$~~ & ~~$\gamma_u$~~ & ~~$\phi_{cr}$ \vspace{1mm}\\\hline  \hline \vspace{-3mm} &&& \\
$-\f{1}{9}$~~	& $\f{1}{3}$& $\f{1}{4}$  & $\f{3}{4}$  \vspace{2mm}\\ 
$>-\f{1}{9}$~~	& $\f{1}{2}$& $\f{1}{3}$  & $\f{2}{3}$  \vspace{2mm}\\ \hline  \hline  &&& \vspace{-2mm}
\end{tabular}
\end{center}
\caption{The critical exponents characterizing the singular behaviour of the generating function of DDP around the multicritical point, depending on the value of $w$.}
\label{table_critexp}
\end{table}

Figure \ref{fig:phasediag} shows a schematic picture of the phase diagram for fixed jump weight $w\geq-\f{1}{9}$. The horizontal line for $t\leq t_c(w)$ is a line of essential singularities and the dashed line is a line of poles, whose slope at the critical point is determined by the crossover exponent $\phi_{cr}$. Both lines meet in the multicritical point. Another picture of the phase diagram for fixed $q=1$ is given in Figure~\ref{fig:critical_t_values}, where the solid line labelled by $t_c^+$ is the line of tri-critical points.

\begin{figure}[htb]
  \centering
  \def\svgwidth{200pt}
  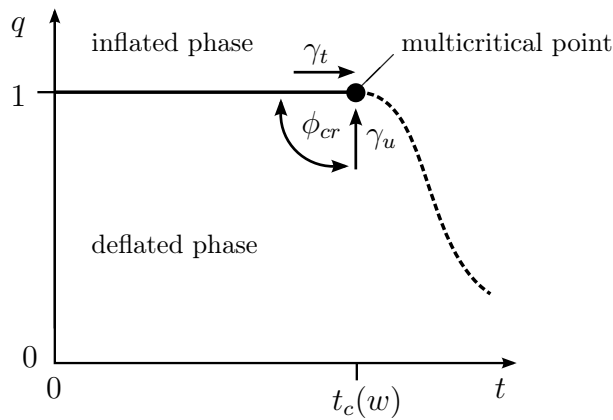
  \caption{Schematic picture of a cut through the 3-dimensional phase diagram of DDP with $t,w$ and $q$ as defined below Eq.\eqref{eq:def_gf_ddp}, along a plane with constant $w\geq-\f{1}{9}$.}

  \label{fig:phasediag}
\end{figure}

\section{Contour integral representation of $\phi(a,q^k t,q)$}
\sectionmark{Contour integral representation}
\label{section:contintrep}

We begin by deriving a contour integral representation for the series defined in Eq.\eqref{eq:def_1_phi_2}. The proof of the following Lemma is a straightforward modification of those given in \cite{Prellberg95} and \cite{Haug15} and will therefore not be carried out here.
\begin{lemma}
For complex $a,t$ with $|\arg(1-a)|<\pi$ and $t\neq 0$, $0<q<1$ and $k\in\mathbb{Z}_{\geq 0}$,
\be
\phi(a,q^k t,q) =\f{A(a,q)}{2\pi i}\int_{C}\f{z^{\f{1}{2}(\log_q(z)+1) - \log_q(t)}}{z^k(z;q)_\infty(a/z;q)_\infty}dz
\label{eq:exact_integral_formula},
\ee
where $A(a,q)=(q;q)_\infty(a;q)_\infty$ and $C$ is a contour from $\infty \exp(-i\psi)$ to $\infty \exp(i\varphi)$, with $(\psi,\varphi)\in\,]0,\pi[^2$, intersecting the real axis at $z=\rho$, where $0<\rho<1$, such that all zeros of $(a/z;q)_\infty$ lie to the left of $C$ -- see Fig. \ref{fig:contour_C}.
\label{lemma:exact_integral_formula}
\end{lemma}
\begin{figure}[htb]
  \centering
  \def\svgwidth{300pt}
  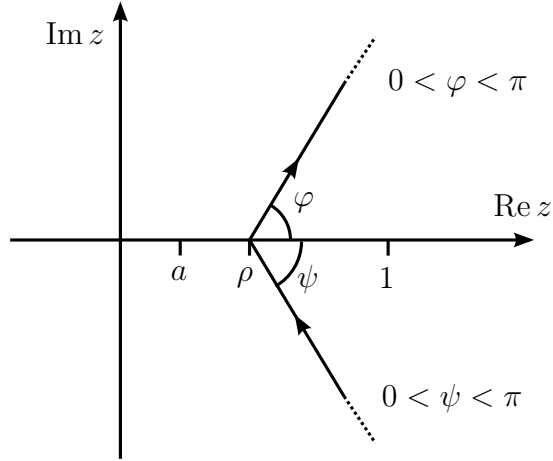
\caption{The contour $C$ used in Eq.\eqref{eq:exact_integral_formula}.}
  \label{fig:contour_C}
\end{figure}

In \cite{Prellberg95} it was shown by applying the Euler-Maclaurin summation formula that for complex $z$ with $\left|\arg(1-z)\right|<\pi$, $0<q<1$ and $m\in\mathbb{N}$,
\be
\ln (z;q)_\infty&=&\f{1}{\ln(q)}\Li_2(z)+\f{1}{2}\ln(1-z)+\nonumber \\
&+&\sum_{n=1}^{m-1}\f{B_{2n}}{(2n)!}(\ln q)^{2n-1}\left(z\f{d}{dz}\right)^{2n-2}\f{z}{1-z}+(\ln q)^{2m-1} R_m(z,q),\qquad
\label{eq:asymptotics_q_product}
\ee
where $B_n$ is the $n$-th Bernoulli number and $\op{Li}_2(z)$ denotes the principal branch of the Euler dilogarithm (\href{http://dlmf.nist.gov/25.12}{$§$25.2} in \cite{NIST}), which for $z \in \mathbb{C}$ is defined as
\be 
\op{Li}_2(z) = -\int_0^z \f{\ln(1-w)}{w} dw,
\label{eq:def_dilog}
\ee
with the principal branch of the logarithm being taken in the integral. Here and in the following, we define $\op{Im}\big(\ln(x)\big)=i\pi$ for any real $x<0$. The remainder term in Eq.\eqref{eq:asymptotics_q_product} satisfies for $m\in \mathbb{N}$,
\be
\left|R_m(z,q)\right| \leq \f{2\,|B_{2m}|}{(2m!)}\int_0^{1}\left|\left(u\f{d}{du}\right)^{2m-1}\f{zu}{1-zu}\right|\f{du}{u}.
\label{eq:general_remainder}
\ee
For $m=1$, evaluating the integral on the rhs of Eq.\eqref{eq:general_remainder} gives the bound
\be 
\left|R_1(z,q)\right| \leq~\f{1}{6\,\sin(\varphi)}\left(\arctan\left(\f{|z|-\cos(\varphi)}{\sin(\varphi)}\right)-\psi\right),
\label{eq:remainder_estimation_R}
\ee 
where $\varphi = \arg(z)$ and $\psi = \varphi\pm \f{\pi}{2}$ for $\phi\lessgtr 0$. 

Using Eq.\eqref{eq:asymptotics_q_product}, we can rewrite Eq.\eqref{eq:exact_integral_formula} as
\BE
\phi(a,q^k t,q) = \f{A(a,q)}{2\pi i} \int_{C} \exp\left(\f 1\epsilon f(z)\right) \f{g_\epsilon(z)}{z^k} dz,
\label{eq:asymptotics_phi_with_remainder}
\EE
where $\epsilon=-\ln(q)$ and for $a\in \mathbb{C}$, $t\in \mathbb{C}\setminus\{0\}$ and $z\in \mathbb{C}\setminus\{0\}$ (resp.$\mathbb{C}\setminus\{0,a,1\}$),
\BE
	f(z)&=&\ln\left(t\right)\ln(z)+\Li_2(z)-\f{1}{2}\ln(z)^2+\Li_2\left(\f{a}{z}\right),
	\label{eq:f}\\	
	g_{\epsilon}(z)&=&\disp \left(\f{z^{2}}{(1-z)(z-a)}\right)^{1/2}\exp\Big[\epsilon \Big(R_1(z,q)+R_1\Big(\f{a}{z},q\Big)\Big)\Big].
	\label{eq:g}
\EE
Note that for $0<a<1$, $f(z)$ and $g_\epsilon(z)$ are real on the segment $a<z<1$, therefore in this case, $f(z^*)=f(z)^*$ and $g_\epsilon(z^*)=g_\epsilon(z)^*$, and $f(z)$ is analytic for $z\in\mathbb{C}\setminus\big(-\infty,a\big]\cup\big[1,\infty\big)$. In order to analyse the integral on the rhs of Eq.\eqref{eq:asymptotics_phi_with_remainder} by means of the saddle point method, we need to further analyse the function $f(z)$. We begin by discussing the location of its saddle points, depending on the values of $a$ and $t$.
\section{Location of the saddle points}
\label{section:location_of_the_saddle_points}

The saddle points of $f(z)$ are the zeros of the derivative
\be 
f^\prime(z) = \f{1}{z}\ln\left(\f{t\,(z-a)}{z^2(1-z)}\right),
\ee
which coincide with the zeros of the polynomial
\be 
s(z) = z^3-z^2+t\,z-t\,a.
\label{eq:saddle_point_equation_GDP}
\ee
Hence, $f(z)$ has (up to multiplicity) three saddle points $z_i~(i=1,2,3)$, which satisfy
\be 
\left. \begin{array}{rcl}
z_1+z_2+z_3 &=&1 \\
z_1z_2+z_2z_3+z_3z_1 &=& t \\
z_1z_2z_3 &=& ta
\end{array}\right\}.
\label{eq:equations_for_zis}
\ee
If the parameter $t$ takes one of the two values
\be 
t_c^{\pm} = \f{1}{8}\left(1+ 18 a-27 a^2 \pm (1-9a) \sqrt{9a^2-10a+1}\right),
\label{eq:critical_t_values_general_a}
\ee
then two saddle points coalesce in one of the points
\be 
z_c^\pm = \f{1}{4}\left(3 a + 1 \pm \sqrt{9 a^2-10 a+1}\right).
\label{eq:points_of_coalescence}
\ee

In Fig.\,\ref{fig:critical_t_values} we show the dependence of the two critical values $t_c^\pm$ as functions of $a$. Concerning the location of the saddle points $z_1,z_2$ and $z_3$ in the complex plane, we distinguish the following 5 cases.

\begin{figure}[htb]
  \centering
  \def\svgwidth{225pt}
  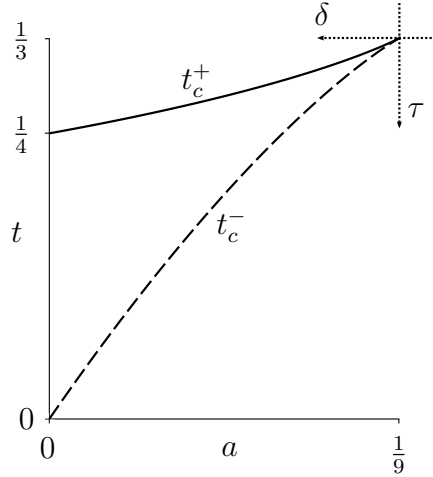
\caption{Plot of the critical values $t_c^{\pm}$ as functions of $a$. The picture also shows the orientation of the natural coordinates $\tau=1/3-t$ and $\delta=1/9-a$ which will be used later on.}
  \label{fig:critical_t_values}
\end{figure}
\begin{enumerate}
\item [(i)] If $a<0$, then two saddle points coalesce for $t=t_c^+<1/4$, while the third one is negative.
\item [(ii)] If $a=0$, then one saddle point is constantly zero while the other two coalesce in $z_c=1/2$ for $t=t_c^+=1/4$. 
\item [(iii)] If $0<a<1/9$, then two saddle points are mutually complex conjugates for $0<t<t_c^-<1/3$ and coalesce on the positive real line for $t=t_c^-$ in the point $z_c^-$, where $a<z_c^-<1/3$. For $t_c^-<t<t_c^+$, all three saddle points are real and for $t=t_c^+<1/3$, two saddle points coalesce in the point $1/3<z_c^+<1/2$. For $t>t_c^+$, again one saddle point is real and the other two are mutually complex conjugates. 
\item [(iv)] If $a=1/9$, then $t_c^-=t_c^+=1/3$, hence all three saddle points coalesce in the same point, $z_c^-=z_c^+=1/3$.
\item [(v)]If $a>1/9$, then there is no saddle point coalescence for $t>0$.
\end{enumerate}

\section{Geometry of the paths of steepest descent}
\label{section:geometry_of_paths_of_steepest_descent}

In this section we are going to discuss the geometry of the paths of steepest descent of $f(z)$ (see e.g. \cite{Flajolet09_8} for a general introduction to the method of steepest descents). To this purpose, we first state the following Lemma, the proof of which relies on basic relations for the Euler dilogarithm (\href{http://dlmf.nist.gov/25.12.E4}{Eq.(25.12.4)} in \cite{NIST}).
\begin{lemma}
For complex $a,t\neq 0$, $0<|\phi|<\pi$ and $\lambda\to 0^+$ or $\lambda \to \infty$,
\be 
f(\lambda\,e^{i\phi}) = -\ln(\lambda)^2 + b \ln(\lambda) - i\,\psi \ln(\lambda)+ \mathcal{O}(1),
\label{eq:asymptotics_f}
\ee
where $b$ denotes the principal branch value of $\ln(a t)$ and $\ln(t)$ with $\ln(-1)=i\pi$ for $\lambda\to 0^+$ and $\lambda\to \infty$, respectively, and $\psi=2\,\phi\mp \pi$ for $\phi\gtrless 0$. 
\label{lemma:asymptotics_f}
\end{lemma}
Note also that for real $x$, real $t>0$ and $0<a<1$,
\be 
\op{Im}(f(x))=\left\{\ba{cl} 
\pi\ln(t/|x|)&\quad (x<0)\vspace{1mm}\\
\pi\ln(x/a)  &\quad(0<x\leq a)\vspace{1mm}\\
0 & \quad (a\leq x\leq 1)\vspace{1mm} \\
-\pi\ln(x)    & \quad(x\geq 1)
\ea.
\right.
\ee
From the sign of the imaginary part of $f^\prime(z)$, we can conclude that paths of steepest descent cannot end at the branch cut of the logarithm. 

Using Lemma \ref{lemma:asymptotics_f}, we prove
\begin{lemma}
For real $a\leq 1/9$~and~$0< t\leq t_c^+(a)$, there exists a continuous curve $c: \mathbb{R} \rightarrow \mathbb{C}$,  with $c(0) = z_{3}$ and $\op{Im}c(\lambda)\gtrless 0$ for $\lambda \gtrless 0$, such that
\be 
\op{Im} f\big(c(\lambda)\big) = 0
\label{eq:vanishing_Impart}
\ee
for $\lambda \in \mathbb{R}$, $\left|c(\lambda)\right| \to \infty$ for $\lambda \to \pm \infty$ and
\be 
\lim_{\lambda\to\pm\infty}\arg\big(c(\lambda)\big) = \pm\f{\pi}{2}.
\label{eq:asymptotics_path_of_steepest_desc}
\ee
\label{lemma:geometry_of_path_of_steepest_descent}
\end{lemma}
\begin{proof}
Assume $a\leq 1/9$~and~$0< t\leq t_c^+(a)$. According to the discussion above, we label the saddle points in such a way that $z_1$ and $z_2$ are mutually complex conjugates for $t<t_c^-(a)$ with $\op{Im}(z_1)>0$ while $z_3$ is real, and $z_2$ coalesces with $z_3$ for $t=t_c^+(a)$. From the asymptotic behaviour of $f(z)$ stated in Lemma \ref{lemma:asymptotics_f}, we can conclude that paths of steepest descent can only end in $z=0$ or at $\infty\exp(\pm i\pi/2)$. Since $f(z^*)=f(z)^*$, it is sufficient to consider the upper half-plane. There are two cases to be distinguished.

\begin{enumerate}
\item \underline{$0 < t < t_c^-(a)$}. In this case, $z_3$ is real while $\op{Im}(z_1)>0$. One of the two paths of steepest descents originating from $z_1$ ends in $z=0$, while the other one ends at infinity. Since paths of steepest descent can only cross in saddle points, it follows that the path of steepest descent emerging from $z_3$ necessarily ends at $\infty\exp(i\pi/2)$. Figure \ref{fig:paths_of_steepest_descent_and_ascent}\,(a) shows an example for this case.

\item \underline{$t_c^-(a) \leq t \leq t_c^+(a)$}. In this case, all three saddle points are real. The path of steepest descent originating from $z_1$ necessarily ends at zero, while the path of steepest ascent originating from $z_2$ ends in the point $z=-t$. Again it follows that the path of steepest descent originating from $z_3$ ends at $\infty\exp(i\pi/2)$. Figure \ref{fig:paths_of_steepest_descent_and_ascent}\,(b) shows an example for $t_c^-(a) < t < t_c^+(a)$ for $a<1/9$ and (c) shows the special case $a=1/9$, for which the three saddle points coalesce.

\end{enumerate}
Note that for $a<0$, $t_c^-(a)<0$, and therefore only the second case is relevant.

 For $0 < t\leq t_c^+(a)$, $\op{Im} f(z_3) = 0$. Since the paths of steepest descent are the contours on which the imaginary part of $f(z)$ is constant, the union of the two paths of steepest originating from $z_3$ and ending at $\infty\exp(\pm i\pi/2)$ has the properties of the curve $c(\lambda)$.
\end{proof}
\begin{figure}[htp]
	\begin{center}	     
\subfigure[$a=0.11$, $t=0.31$]{\def\svgwidth{175pt}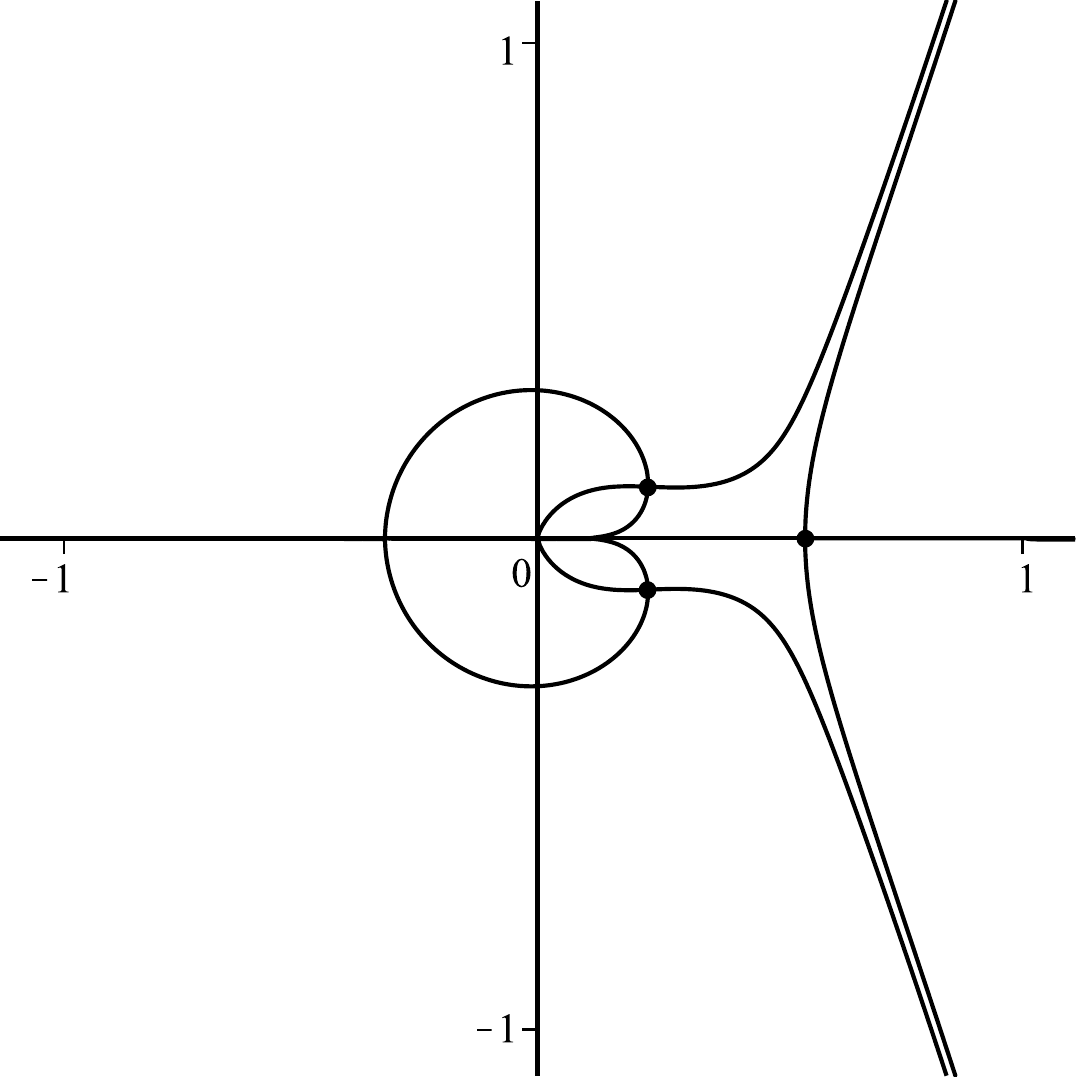}
\hspace{0.5cm}  
\subfigure[$a=0.11$, $t=0.3317$]{\def\svgwidth{175pt}
  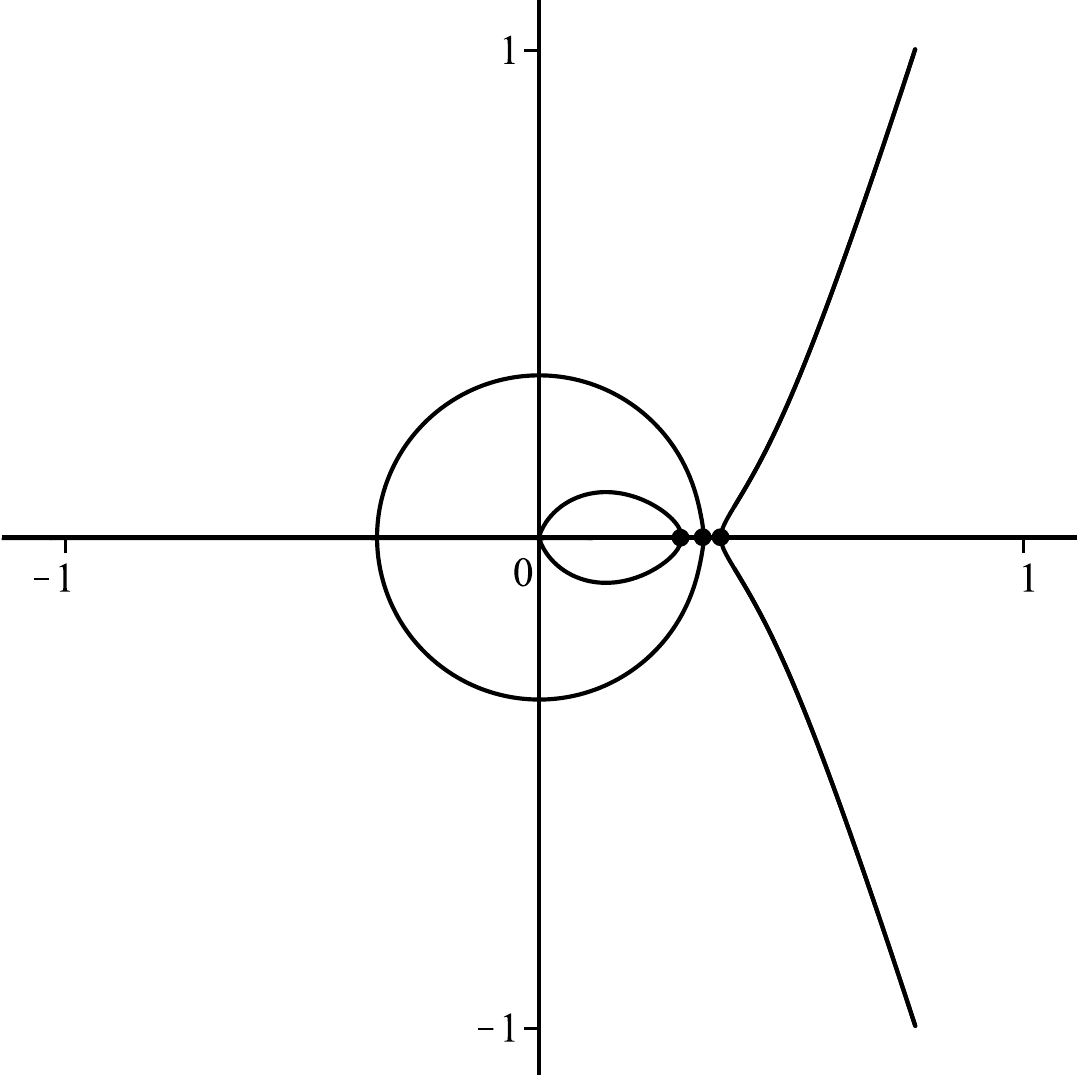}
\hspace{0.5cm}  
\subfigure[$a=1/9$, $t=1/3$]{\def\svgwidth{175pt}
  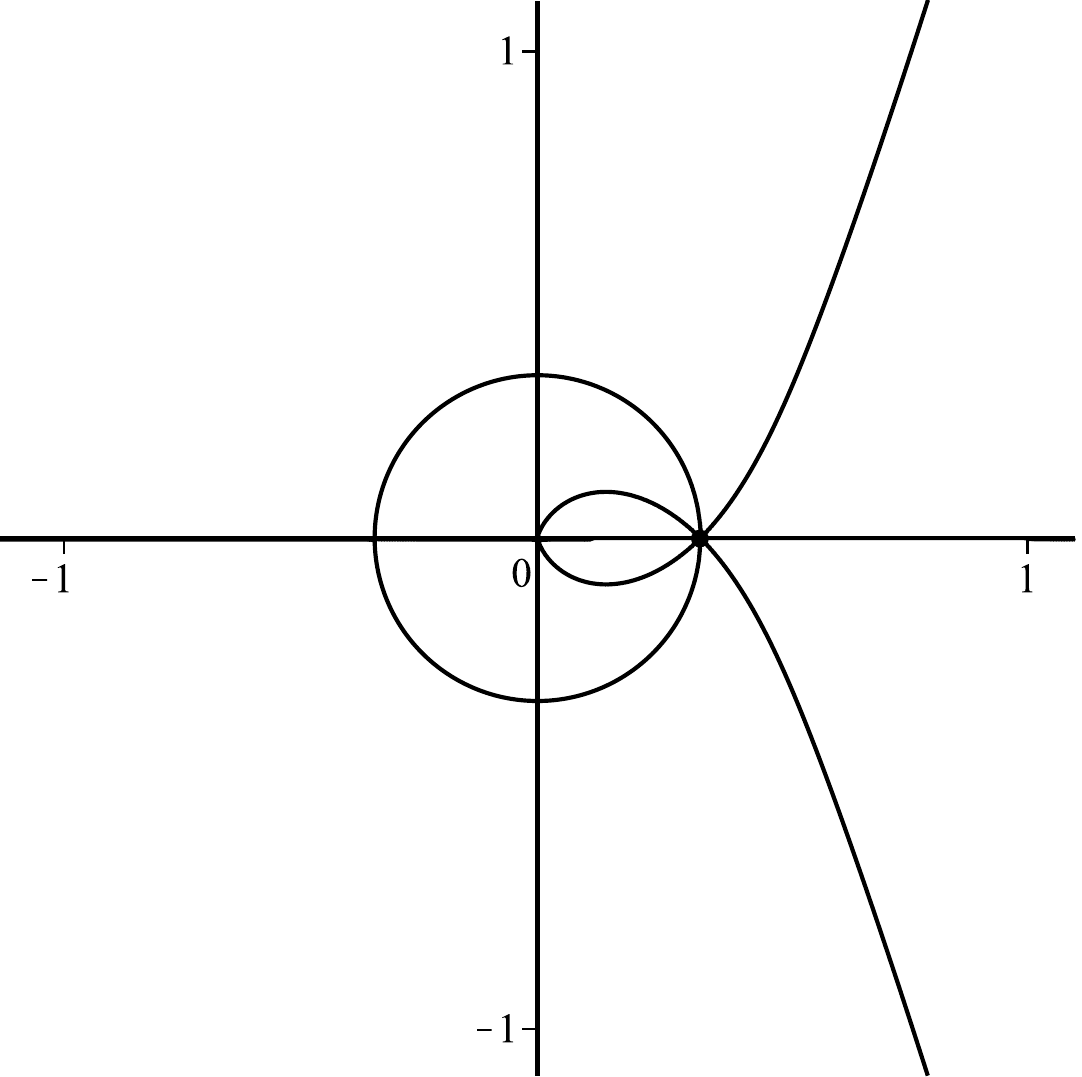}
\hspace{0.5cm}  
\subfigure[$a=0.11$, $t=0.34$]{\def\svgwidth{175pt}
  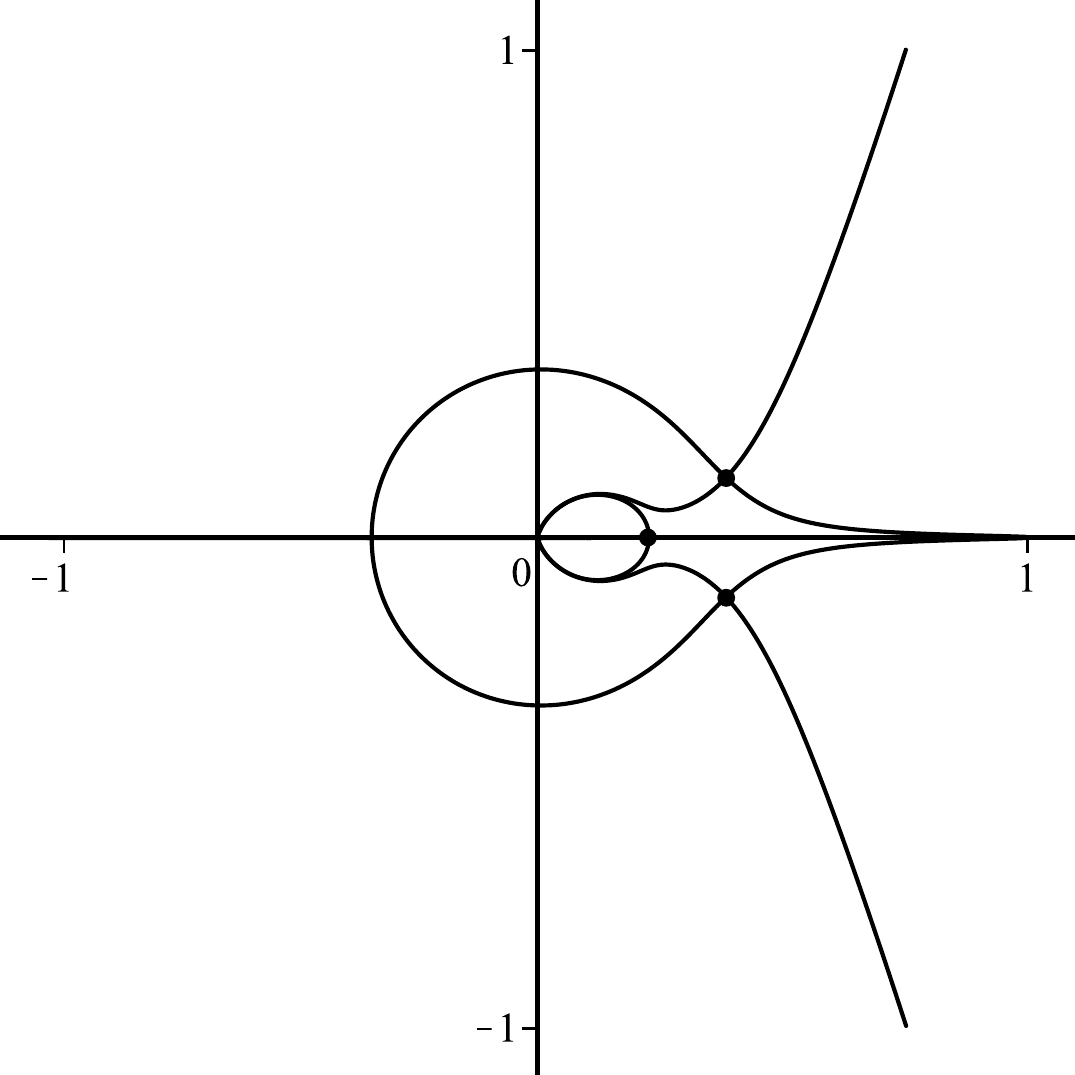}
	\end{center}
\caption{The saddle points of the function $f(z)$ defined in Eq.\eqref{eq:f} (marked by small black dots) and the paths of steepest descent and ascent originating from them.}
	\label{fig:paths_of_steepest_descent_and_ascent}
\end{figure}

\section{Transformation of $f(z)$ into a canonical form}
\label{section:canonicaltrafo}
As we discussed in Section \ref{section:location_of_the_saddle_points}, for $a=1/9$, the three saddle points of $f(z)$ coalesce in the point $z_c=1/3$ for $t=t_c^+(1/9)=1/3$. We now define the natural coordinates
\be 
\tau = \f13-t~~\text{and}~~\delta = \f19-a,
\ee
and consider $f(z)$ and $g_\epsilon(z)$ as functions of $z$, $\tau$ and $\delta$ from now on.

Theorem 1 from \cite{Ursell72} states that there is a map $z \mapsto u(z),$ such that
\be 
f(z) = \f{1}{4}\,u^4 - \alpha\,u^2 - \beta\,u + \gamma = p(u),
\label{eq:transformation_formula}
\ee
which is regular and bijective if $z-z_c$, $\tau$ and $\delta$ are sufficiently close to zero. The coefficients $\alpha$, $\beta$ and $\gamma$ are regular functions of $\tau$ and $\delta$. Thus,
\be
\alpha = \sum_{j,k=0}^\infty \alpha_{j,k} \tau^j \delta^k \quad ;\quad \beta = \sum_{j,k=0}^\infty \beta_{j,k}\tau^j \delta^k \quad ;\quad 
\gamma = \sum_{j,k=0}^\infty \gamma_{j,k} \tau^j \delta^k.
\label{eq:expansions_alphabetagamma}
\ee
Note that $\alpha$ and $\beta$ are not unique, since the form of the rhs of Eq.\eqref{eq:transformation_formula} is invariant under rotations about the angle $k \pi/2$, where $k\in\mathbb{Z}$. 

We denote the three saddle points of the polynomial $p(u)$ by $u_1,u_2$ and $u_3$, hence
\be
p'(u) = u^3-2\alpha u -\beta = (u-u_1)(u-u_2)(u-u_3).
\label{eq:derivative_p}
\ee
From Eq.\eqref{eq:derivative_p} it follows that
\begin{subequations}
\be
u_1+u_2+u_3 &=&0, \label{eq:basic_eqn_us1}\\
u_1u_2+u_2u_3+u_3u_1 + 2\alpha &=& 0,\label{eq:basic_eqn_us2}\\
u_1u_2u_3 -\beta &=& 0.\label{eq:basic_eqn_us3}
\ee
\label{eq:basic_eqns_us}
\end{subequations}
As a necessary condition for the transformation defined in Eq.\eqref{eq:transformation_formula} to be regular, the saddle points of of $f(z)$ need to be mapped onto the saddle points of $p(u)$. For $\tau=\delta=0$, the three saddle points of $f(z)$ coalesce and it follows from Eq.\eqref{eq:basic_eqn_us1} that $u_1=u_2=u_3=0$. With this we obtain from Eqs.(\ref{eq:basic_eqns_us}b-c) that $\alpha_{0,0}=\beta_{0,0}=0$. 

If we label the saddle points of $p(u)$ such that $u(z_j) = u_{j}$ for $j=1,2$ and $3$, then from Eq.\eqref{eq:transformation_formula} it follows by differentiating twice that

\be 
\left.\f{dz}{du}\right|_{u=u_i} = \left(\f{3 u_i^2-2\alpha}{f^{\prime\prime}(z_i)}\right)^{1/2},
\label{eq:Jacobian}
\ee
for $(\tau,\delta)\neq(0,0)$. For $\tau=\delta=0$, we get by taking higher derivatives that
\be 
\left.\f{dz}{du}\right|_{u=0}=\f{2^{1/4}}{3}~\text{,}~~\left.\f{d^2 z}{du^2}\right|_{u=0}=\f{2^{1/2}}{3}~~\text{and}~~\left.\f{d^3 z}{du^3}\right|_{u=0}=\f{2^{3/4}}{6}.
\label{eq:Jacobian1}
\ee
Since $u_j^3=2\alpha u_j+\beta$ for $j=1,2$ and 3, we have
\be 
f(z_j) = -\f{1}{2}\,\alpha\,u_{j}^2-\f{3}{4}\,\beta\,u_{j} +\gamma.
\label{eq:fatsaddles}
\ee
Using Eqs.\eqref{eq:basic_eqns_us} and Eq.\eqref{eq:fatsaddles}, we obtain the following set of equations, where $\Sigma^{(k)}=\sum_{j=1}^3 f(z_j)^k$ for $k=1,2$ and 3.
 \def\thesubequation{\arabic{equation}\roman{subequation}}
\begin{subequations}
\be 
3\gamma-2\alpha^2&=&\Sigma^{(1)},\quad\label{eq:fis_eq1}\\ 
3\gamma^2+\f{9}{2}\alpha\beta^2-4\alpha^2\gamma+2\alpha^4&=&\Sigma^{(2)},\quad\label{eq:fis_eq2}\\
3\gamma^3-\f{81}{64}\beta^4+\f{27}{2}\alpha\beta^2\gamma-6\alpha^2\gamma^2-\f{51}{4}\,\alpha^3\beta^2+6\,\alpha^4\gamma-2\alpha^6&=&\Sigma^{(3)}.\quad\label{eq:fis_eq3}
\ee
\label{eq:fis_eqs}
\end{subequations}
From Eqs.\eqref{eq:fis_eqs}, we derive
\begin{subequations}
\be 
2\alpha^4+\f{27}{2}\alpha\beta^2 &=& 3 \Sigma^{(2)}-\Big(\Sigma^{(1)}\Big)^2,\label{eq:fis_eq_derived1}\\
\alpha^6-\f{135}{8}\alpha^3\beta^2-\f{729}{128}\beta^4 &=& \Big(\Sigma^{(1)}\Big)^3+\f{9}{2}\Big(\Sigma^{(3)}-\Sigma^{(1)}\Sigma^{(2)}\Big)\label{eq:fis_eq_derived2}.
\ee
\label{eq:fis_eqs_derived}
\end{subequations}
With Eqs.\eqref{eq:fis_eqs} and \eqref{eq:fis_eqs_derived} we are now going to calculate the leading coefficients of the power series expansions \eqref{eq:expansions_alphabetagamma}. We will begin by considering the cases $\delta=0$ and $\tau=0$ separately.

\subsection{Coefficient asymptotics for $\tau\to 0$ and $\delta=0$}

From the above discussion we know that for $\delta=0$ and $\tau\to 0$, $\alpha\sim\alpha_{r_\alpha,0}\tau^{r_\alpha}$ and $\beta \sim\beta_{r_\beta,0}\tau^{r_\beta}$, where $r_\alpha,r_\beta \in \mathbb{N}$ and $\alpha_{r_\alpha,0},\beta_{r_\beta,0}\neq 0$.

To determine $r_\beta$, we take the third derivative of Eq.\eqref{eq:transformation_formula} with respect to $u$ and insert the saddle point values. This gives us
\be 
f^{\prime\prime\prime}(z_j)\left(\left.\f{dz}{du}\right|_{u_j}\right)^3+3f^{\prime\prime}(z_j)\left.\f{dz}{du}\right|_{u_j}\left.\f{d^2 z}{du^2}\right|_{u_j} = 6\, u_j
\label{eq:equation_for_ui}
\ee
for $j=1,2$ and $3$. Expanding both $f^{\prime\prime}(z_j)$ and  $f^{\prime\prime\prime}(z_j)$ as series in $\tau$ shows that for $\tau\to 0^+$, $f^{\prime\prime}(z_j) = o(\tau^{1/3})$. Moreover, for $k=1,2$ and 3,
\be 
f^{\prime\prime\prime}(z_k)= c_0\,\exp\left(\f{2 k\pi i}{3}\right)\,\tau^{1/3} + \mathcal{O}(\tau^{2/3}),
\ee
where $c_0=81\cd6^{1/3}$ and from this it follows together with \eqref{eq:Jacobian1} that 
\be 
u_k = u_0\exp\left(\f{2 k\pi i}{3}\right)\,\tau^{1/3} + \mathcal{O}(\tau^{2/3}),
\label{eq:asymptotics_ui}
\ee
where $u_0=6^{1/3}/2^{1/4}$. From Eq.\eqref{eq:basic_eqn_us3} we therefore conclude that $r_\beta=1$. With this we can now determine $r_\alpha$. From Eq.\eqref{eq:fis_eq_derived1}, we obtain for $\delta=0$ and $\tau\to 0$
\be
2\,\alpha^4+\f{27}{2}\alpha \beta^2 = \f{6561}{320}\tau^4 + \mathcal{O}(\tau^5).
\label{eq:fis_eq5}
\ee
It follows by the following dominant balance argument that $r_\alpha=2$. There are three possibilities to be distinguished. The first possibility is that $\alpha^4=o(\alpha\beta^2)$, from which it would follow that $r_\alpha+2 =4$, hence $r_\alpha=2$. The second possibility is that $\alpha\beta^2=o(\alpha^4)$, from which it would follow that $4\,r_\alpha=4$, hence $r_\alpha=1$. However, this would mean that $r_\alpha+2=3$, which stands in contradiction to the assumption that $\alpha\beta^2=o(\alpha^4)$. The third possibility is that the leading terms of $\alpha^4$ and $\alpha\beta^2$ cancel each other. In that case, $4\,r_\alpha=r_\alpha+2$, hence $r_\alpha=2/3$, which is impossible. We conclude that $r_\alpha=2$.

Expanding the rhs of Eq.\eqref{eq:fis_eq1} for $\delta=0$ in $\tau$, we get
\be 
\gamma_{0,0}=2\Li_2\left(\f13\right)+\f12\ln(3)^2,
\label{eq:gamma_0}
\ee
and using Eq.\eqref{eq:fis_eq5} and expanding Eq.(\ref{eq:fis_eqs_derived}b) in $\tau$ for $\delta=0$, we obtain $\alpha_{1,0}\beta_{1,0}^2$ and $\beta_{1,0}^4$. Chosing the real positive root for $\beta_{1,0}$, we arrive at
\be 
\alpha_{2,0} = \f{27\sqrt{2}}{320}~~\text{and}~~\beta_{1,0} = 3\sqrt[4]{2}.
\label{eq:lowest_coeffs_alphabeta_tau}
\ee
One can easily calculate further expansion coefficients, but here we will only give the results for the leading orders.
\subsection{Coefficient asymptotics for $\delta\to 0$ and $\tau=0$}
It follows from an argument analogous to the one given in the previous subsection that for $\tau=0$ and $\delta\to 0$, $\alpha\sim\alpha_{0,1}\delta$ and $\beta\sim \beta_{0,1} \delta$, where $\alpha_{0,1},\beta_{0,1}\neq 0$.

Again using Eqs.(\ref{eq:fis_eqs_derived}a-b), we obtain the values for $\alpha_{0,1}\beta_{0,1}^2$ and $\beta_{0,1}^4$. Since in the previous subsection we have chosen the positive root for $\beta_{1,0}$, we need to make sure to chose the correct root for $\beta_{0,1}$. Setting $\delta=-\tau$ and expanding the rhs of Eq.\eqref{eq:fis_eq_derived2} in $\tau$, we get $\beta_{1,0}-\beta_{0,1} = (15/2)\sqrt[4]{2}$. From this it follows that we need to choose the real negative root for $\beta_{0,1}$. We obtain
\be 
\alpha_{0,1} = \f{27\sqrt{2}}{8}~~\text{and}~~\beta_{0,1} = -\f{9}{2} \sqrt[4]{2}.
\label{eq:lowest_coeffs_alphabeta_delta}
\ee
Combining Eqs.\eqref{eq:lowest_coeffs_alphabeta_tau} and \eqref{eq:lowest_coeffs_alphabeta_delta}, we get that for $(\tau,\delta)\to(0,0)$,
\be 
\alpha \sim  \f{27\sqrt{2}}{8}\left(\delta + \f{1}{40} \tau^2 \right)~~\text{and}~~\beta  \sim  3 \sqrt[4]{2}\left(\tau - \f{3}{2} \delta \right).
\label{eq:double_series_alpha_beta}
\ee
\section{Asymptotics of $\phi(a,q^k t,q)$}
\label{section:asyofphi}

It follows from Lemma \ref{lemma:geometry_of_path_of_steepest_descent} together with Cauchy's theorem that for $0<t\leq 1/3$ and $a\leq 1/9$, we can replace the integration contour in Eq.\eqref{eq:asymptotics_phi_with_remainder} by a contour $C_0$ originating from $\infty\exp(-i\f{\pi}{2})$, passing through the real valued saddle point $z_3$ of $f(z)$ and ending at $\infty\exp(i\f{\pi}{2})$, such that $\op{Im} f(z)=0$ on this contour and $\op{Re}f(z)$ is maximal at $z_3$. 

The correction due to restricting $C_0$ to the central part $C_0'$ on which the transformation defined in Eq.\eqref{eq:transformation_formula} is regular decays exponentially in the limit $\epsilon\to 0^+$. The segment $u(C_0^\prime)$ is the central part of the contour given by the union of the two paths of steepest descent of $p(u)$ ending at $\infty\exp(\pm i\pi/4)$. Extending the integration to the complete contour, we obtain
\be 
\phi(a,q^k t,q) = \f{A(a,q)}{2\pi i}\int_{c_-\infty}^{c_+\infty} \exp\left(\f{1}{\epsilon} p(u)\right) G^{(k)}(u)\Big(1+\mathcal{O}(\epsilon)\Big) du,
\label{eq:transformation_integral_a<1/9}
\ee
where $c_{\pm}=\exp\left(\pm i\f{\pi}{4}\right)$, and 
\be 
G^{(k)}(u) = \f{g_0(z(u))}{z(u)^k}\,\f{dz}{du}.
\label{eq:g_times_Jacobian}
\ee

In order to calculate the leading asymptotic contribution to $\phi(a,q^k t,q)$, we use the ansatz \cite{Ursell72} 
\be 
G^{(k)}(u)=P^{(k)} + u Q^{(k)} + u^2 R^{(k)} + (u^3-2\alpha u-\beta) S^{(k)}(u),
\label{eq:expansion_Ansatz_G}
\ee
where $S^{(k)}(u)$ is an analytic function of $u, \tau$ and $\delta$ and  $P^{(k)},~Q^{(k)}$ and $R^{(k)}$ are analytic functions of $\tau$ and $\delta$. Inserting the saddle point values into Eq.\eqref{eq:expansion_Ansatz_G}, we get for $j=1,2$ and $3$,
\be 
G^{(k)}(u_j)&=&P^{(k)} + u_j Q^{(k)} + u_j^2 R^{(k)}.
\label{eq:G_uis}
\ee
Evaluating Eq.\eqref{eq:expansion_Ansatz_G} and the first and second derivative with respect to $u$ thereof at $u=0$ for $\tau=\delta=0$ gives together with \eqref{eq:Jacobian1} for $\tau=\delta=0$,
\be 
\ba{ccccccc}
P^{(0)}&=&\disp\f{2^{1/4}\sqrt{3}}{6},\quad Q^{(0)}&=&\disp\f{\sqrt{6}}{4}, \quad R^{(0)}&=&\disp \f{5\,2^{3/4}\sqrt{3}}{24}, \vspace{2.5mm}\\
P^{(1)}&=&\disp\f{2^{1/4}\sqrt{3}}{2},\quad Q^{(1)}&=&\disp\f{\sqrt{6}}{4}, \quad R^{(1)}&=&\disp \f{2^{3/4}\sqrt{3}}{8}. \\
\ea
\ee

From Eq.\eqref{eq:generalized_Airy} we define the functions
\be 
\Theta_4^{(1)}(x,y) = \f{\pt}{\pt x}\Theta_4(x,y)~~\text{and}~~\Theta_4^{(2)}(x,y) = \f{\pt}{\pt y}\Theta_4(x,y).
\ee
Note that $\Theta_4(x,y)$ is related to the Pearcey integral
\be 
\mathcal{P}(x,y) = 2\exp\left(\f{i\pi}{8}\right)\int_{0}^\infty \exp\left(-u^4-y u^2\right)\cos(x u) du,
\ee
the asymptotics of which has been studied in \cite{Paris91}, via the formula
\be 
\mathcal{P}(x,y) = \f{\sqrt{2}\,\pi}{\exp\left(i\f{\pi}{8}\right)} \left[\Theta_4\left(\f{1-i}{2}x,\f{iy}{2}\right)+i\,\Theta_4\left(\f{1+i}{2}x,\f{-iy}{2}\right)\right].
\ee
Applying Theorem 2 from \cite{Ursell72}, we can now formulate the following result.
\begin{proposition}
For $k\in\mathbb{Z}_{\geq 0}$, there exist constants $d_a,d_t>0$ such that for $a \in \big[\f{1}{9}-d_a,\f{1}{9}+d_a\big]$ and $t\in \big[\f{1}{3}-d_t,\f{1}{3}+d_t\big]$ and $q=e^{-\epsilon}\to 1^-$, we have 
\begin{multline}
\phi(a,q^k t,q) = A(a,q)\,\exp\left(\f\gamma\epsilon\right)\,\Bigl[P^{(k)}\,\epsilon^{1/4}\,\Theta_4\left(\f{\beta}{\epsilon^{3/4}},\f{\alpha}{\epsilon^{1/2}}\right)-\\
 - Q^{(k)}\,\epsilon^{1/2} \Theta_4^{(1)}\left(\f{\beta}{\epsilon^{3/4}},\f{\alpha}{\epsilon^{1/2}}\right)- R^{(k)}\, \epsilon^{3/4}\,\Theta_4^{(2)}\left(\f{\beta}{\epsilon^{3/4}},\f{\alpha}{\epsilon^{1/2}}\right)\Bigl]\big(1+\mathcal{O}(\epsilon)\big),
\label{eq:asymptotics_1_phi_2_a,t<1}
\end{multline}
uniformly, where the coefficients $\alpha,\beta,\gamma$ and $P^{(k)},Q^{(k)},R^{(k)}$ are regular functions of $a$ and $t$ and $A(a,q) = (q;q)_\infty(a;q)_\infty$. 
\label{prop:asymptotics_a<=1/9}
\end{proposition}

Substituting Eq.\eqref{eq:asymptotics_1_phi_2_a,t<1} into Eq.\eqref{eq:gf_generalized_DP} for $k=0$ and $1$ and $a=-w$, we obtain the scaling behaviour of $G(w,t,q)$ around the critical point $(w,t,q)=(-\f{1}{9},\f{1}{3},1)$ as stated in Theorem \ref{thm:scaling_around_multicritical_point}. In Fig. \ref{fig:scaling function}, we show the convergence of the asymptotic approximation of $F(s)$ obtained by rearranging Eq.\eqref{eq:scaling_relation_delta=0} against the exact scaling function.

\begin{figure}[hbt]
\centering
\includegraphics[width=0.6\textwidth]{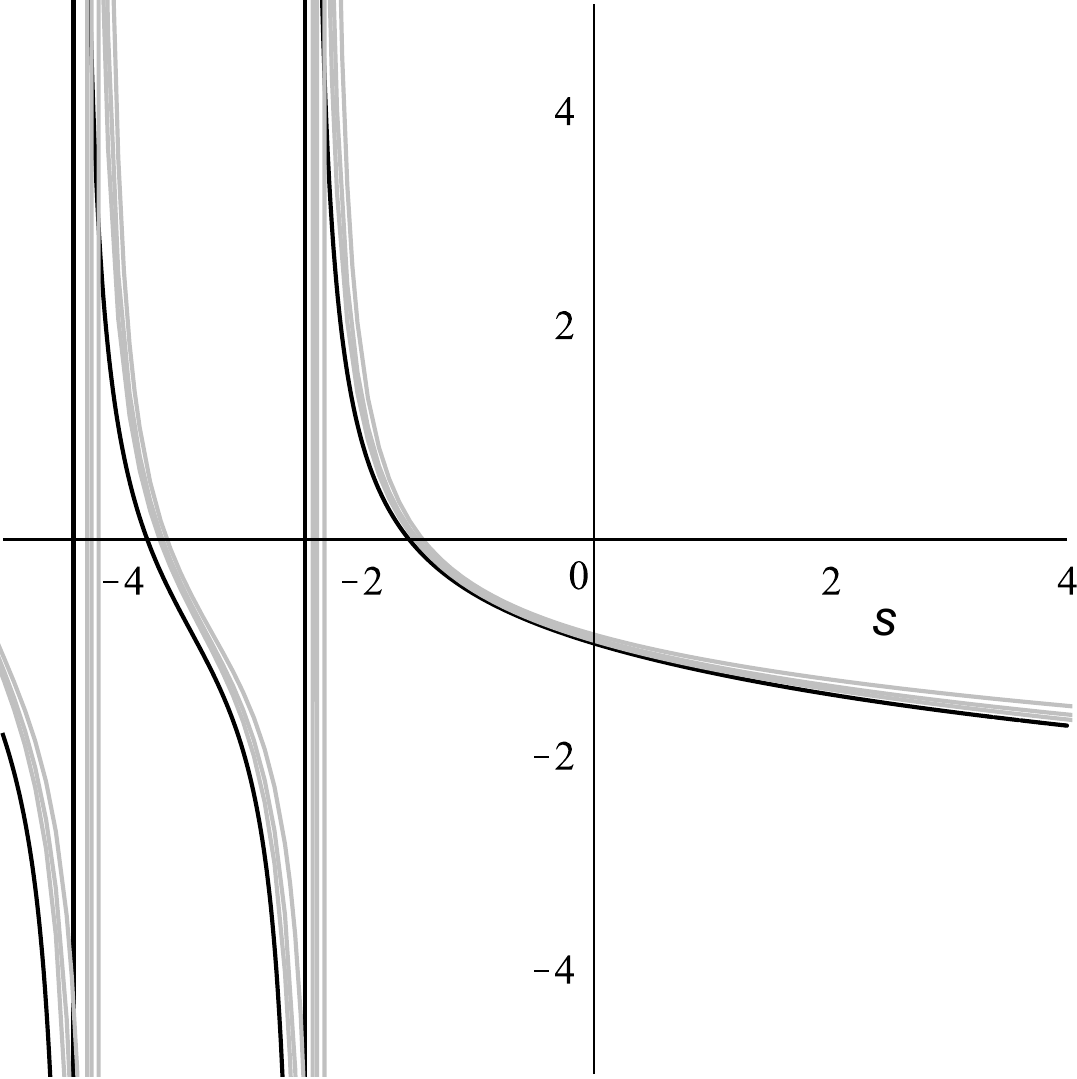}
\caption{Plot of the scaling function $F\big(\sqrt[4]{2}\,s\big)=\Phi\big(\sqrt[4]{2}\,s,0\big)$ (black) and the asymptotic approximation obtained from rearranging Eq.\eqref{eq:scaling_relation_delta=0} for $\epsilon=10^{-4},10^{-5},10^{-6}$ (gray, the smallest value corresponds to the closest approximation).}
  \label{fig:scaling function}
\end{figure}

As discussed in Lemma \ref{lemma:geometry_of_path_of_steepest_descent}, for $a<\f{1}{9}$ and $0<t<t_c^+(a)$, the integration contour $C$ used in Eq.\eqref{eq:exact_integral_formula} can be deformed such that it consists of two paths of steepest descent, connecting a saddle point on the real axis with infinity, and the asymptotics of $\phi(a,q^k t,q)$ can be obtained via the ordinary method of steepest descent. According to Section \ref{section:location_of_the_saddle_points}, the relevant saddle point coalesces with another saddle point for $t=t_c^+(a)<\f{1}{3}$. At this point, $\phi(a,q^kt,q)$ can be approximated in terms of Airy functions, with the special case $a=0$ having been treated in \cite{Haug15}.


\end{document}

%% file: functional_equation_illustration.pdf_tex
\begingroup%
  \makeatletter%
  \providecommand\color[2][]{%
    \errmessage{(Inkscape) Color is used for the text in Inkscape, but the package 'color.sty' is not loaded}%
    \renewcommand\color[2][]{}%
  }%
  \providecommand\transparent[1]{%
    \errmessage{(Inkscape) Transparency is used (non-zero) for the text in Inkscape, but the package 'transparent.sty' is not loaded}%
    \renewcommand\transparent[1]{}%
  }%
  \providecommand\rotatebox[2]{#2}%
  \ifx\svgwidth\undefined%
    \setlength{\unitlength}{210.87499937bp}%
    \ifx\svgscale\undefined%
      \relax%
    \else%
      \setlength{\unitlength}{\unitlength * \real{\svgscale}}%
    \fi%
  \else%
    \setlength{\unitlength}{\svgwidth}%
  \fi%
  \global\let\svgwidth\undefined%
  \global\let\svgscale\undefined%
  \makeatother%
  \begin{picture}(1,0.10)%
    \put(0,0){\includegraphics[width=\unitlength]{functional_equation_illustration.pdf}}%
    \put(0.29764077,0.015){\color[rgb]{0,0,0}\makebox(0,0)[lb]{$+$}}%
    \put(0.62872556,0.015){\color[rgb]{0,0,0}\makebox(0,0)[lb]{$+$}}%
    \put(0.15,0.015){\makebox(0,0)[lb]{\smash{$=$}}}%
  \end{picture}%
\endgroup%

%% file: phasediag.pdf_tex
\begingroup%
  \makeatletter%
  \providecommand\color[2][]{%
    \errmessage{(Inkscape) Color is used for the text in Inkscape, but the package 'color.sty' is not loaded}%
    \renewcommand\color[2][]{}%
  }%
  \providecommand\transparent[1]{%
    \errmessage{(Inkscape) Transparency is used (non-zero) for the text in Inkscape, but the package 'transparent.sty' is not loaded}%
    \renewcommand\transparent[1]{}%
  }%
  \providecommand\rotatebox[2]{#2}%
  \ifx\svgwidth\undefined%
    \setlength{\unitlength}{367.96322004bp}%
    \ifx\svgscale\undefined%
      \relax%
    \else%
      \setlength{\unitlength}{\unitlength * \real{\svgscale}}%
    \fi%
  \else%
    \setlength{\unitlength}{\svgwidth}%
  \fi%
  \global\let\svgwidth\undefined%
  \global\let\svgscale\undefined%
  \makeatother%
  \begin{picture}(1,0.76172539)%
    \put(0,0){\includegraphics[width=\unitlength]{phasediag.pdf}}%
    \put(0.02802149,0.61607872){\color[rgb]{0,0,0}\makebox(0,0)[lt]{\begin{minipage}{0.94191589\unitlength}\raggedright $1$\end{minipage}}}%
    \put(0.028937,0.74495898){\color[rgb]{0,0,0}\makebox(0,0)[lt]{\begin{minipage}{0.19035234\unitlength}\raggedright $q$\end{minipage}}}%
    \put(0.63040718,0.0489639){\color[rgb]{0,0,0}\makebox(0,0)[lt]{\begin{minipage}{1.06334752\unitlength}\raggedright $t_c(w)$\end{minipage}}}%
    \put(0.93477791,0.02676937){\color[rgb]{0,0,0}\makebox(0,0)[lb]{\smash{$t$}}}%
    \put(0.0950392,0.06475949){\color[rgb]{0,0,0}\makebox(0,0)[lt]{\begin{minipage}{0.38398662\unitlength}\raggedright $0$\end{minipage}}}%
    \put(0.05,0.12262628){\color[rgb]{0,0,0}\makebox(0,0)[lt]{\begin{minipage}{0.38398662\unitlength}\raggedright $0$\end{minipage}}}%
    \put(0.69371364,0.53209235){\color[rgb]{0,0,0}\makebox(0,0)[lt]{\begin{minipage}{1.02724626\unitlength}\raggedright $\gamma_u$\end{minipage}}}%
    \put(0.57673978,0.69373077){\color[rgb]{0,0,0}\makebox(0,0)[lt]{\begin{minipage}{0.94847975\unitlength}\raggedright $\gamma_t$\end{minipage}}}%
    \put(0.57279818,0.56723095){\color[rgb]{0,0,0}\makebox(0,0)[lt]{\begin{minipage}{0.3043404\unitlength}\raggedright $\phi_{cr}$\end{minipage}}}%
    \put(0.67,0.6){\line(1,1){0.075}}
    \put(0.76,0.68){\footnotesize multicritical point}
    \put(0.18,0.3){\footnotesize deflated phase}
    \put(0.18,0.68){\footnotesize inflated phase}
  \end{picture}%
\endgroup%

%% file: contour.pdf_tex
\begingroup%
  \makeatletter%
  \providecommand\color[2][]{%
    \errmessage{(Inkscape) Color is used for the text in Inkscape, but the package 'color.sty' is not loaded}%
    \renewcommand\color[2][]{}%
  }%
  \providecommand\transparent[1]{%
    \errmessage{(Inkscape) Transparency is used (non-zero) for the text in Inkscape, but the package 'transparent.sty' is not loaded}%
    \renewcommand\transparent[1]{}%
  }%
  \providecommand\rotatebox[2]{#2}%
  \ifx\svgwidth\undefined%
    \setlength{\unitlength}{203.52766113bp}%
    \ifx\svgscale\undefined%
      \relax%
    \else%
      \setlength{\unitlength}{\unitlength * \real{\svgscale}}%
    \fi%
  \else%
    \setlength{\unitlength}{\svgwidth}%
  \fi%
  \global\let\svgwidth\undefined%
  \global\let\svgscale\undefined%
  \makeatother%
  \begin{picture}(1,0.57590668)%
    \put(0,0){\includegraphics[width=\unitlength]{contour.pdf}}%
    \put(0.6034125,0.30525408){\color[rgb]{0,0,0}\makebox(0,0)[lb]{\smash{$\operatorname{Re} z$}}}%
    \put(0.46020935,0.21471511){\color[rgb]{0,0,0}\makebox(0,0)[lb]{\smash{1}}}%
    \put(0.0445191,0.52208207){\color[rgb]{0,0,0}\makebox(0,0)[lb]{\smash{$\operatorname{Im} z$}}}%
    \put(0.47169019,0.46384379){\color[rgb]{0,0,0}\makebox(0,0)[lb]{\smash{$0 < \varphi < \pi$}}}%
    \put(0.46382885,0.06684621){\color[rgb]{0,0,0}\makebox(0,0)[lb]{\smash{$0 < \psi < \pi$}}}%
    \put(0.28301805,0.224073){\color[rgb]{0,0,0}\makebox(0,0)[lb]{\smash{$\rho$}}}%
    \put(0.20070108,0.2418963){\color[rgb]{0,0,0}\makebox(0,0)[lt]{\begin{minipage}{0.20770569\unitlength}\raggedright  $a$\end{minipage}}}%
    \put(0.35348677,0.31736349){\color[rgb]{0,0,0}\makebox(0,0)[lb]{\smash{$\varphi$}}}%
    \put(0.35803689,0.2161193){\color[rgb]{0,0,0}\makebox(0,0)[lb]{\smash{$\psi$}}}%
  \end{picture}%
\endgroup%

%% file: critical_t_values.pdf_tex
\begingroup%
  \makeatletter%
  \providecommand\color[2][]{%
    \errmessage{(Inkscape) Color is used for the text in Inkscape, but the package 'color.sty' is not loaded}%
    \renewcommand\color[2][]{}%
  }%
  \providecommand\transparent[1]{%
    \errmessage{(Inkscape) Transparency is used (non-zero) for the text in Inkscape, but the package 'transparent.sty' is not loaded}%
    \renewcommand\transparent[1]{}%
  }%
  \providecommand\rotatebox[2]{#2}%
  \ifx\svgwidth\undefined%
    \setlength{\unitlength}{438.76396875bp}%
    \ifx\svgscale\undefined%
      \relax%
    \else%
      \setlength{\unitlength}{\unitlength * \real{\svgscale}}%
    \fi%
  \else%
    \setlength{\unitlength}{\svgwidth}%
  \fi%
  \global\let\svgwidth\undefined%
  \global\let\svgscale\undefined%
  \makeatother%
  \begin{picture}(1,0.86300671)%
    \put(0,0){\includegraphics[width=\unitlength]{critical_t_values.pdf}}%
    \put(0.42807467,0.07896481){\color[rgb]{0,0,0}\makebox(0,0)[lt]{\begin{minipage}{0.1323083\unitlength}\raggedright $a$\end{minipage}}}%
    \put(0.12565068,0.08696624){\color[rgb]{0,0,0}\makebox(0,0)[lt]{\begin{minipage}{0.1323083\unitlength}\raggedright 0\end{minipage}}}%
    \put(0.09056129,0.13632886){\color[rgb]{0,0,0}\makebox(0,0)[lt]{\begin{minipage}{0.1323083\unitlength}\raggedright 0\end{minipage}}}%
    \put(0.73764654,0.6473621){\color[rgb]{0,0,0}\makebox(0,0)[lt]{\begin{minipage}{0.18036542\unitlength}\raggedright $\tau$\end{minipage}}}%
    \put(0.58250998,0.81650372){\color[rgb]{0,0,0}\makebox(0,0)[lt]{\begin{minipage}{0.19238977\unitlength}\raggedright $\delta$\end{minipage}}}%
    \put(0.41961143,0.46965232){\color[rgb]{0,0,0}\makebox(0,0)[lt]{\begin{minipage}{0.29478523\unitlength}\raggedright $t_c^-$\end{minipage}}}%
    \put(0.36133033,0.71308774){\color[rgb]{0,0,0}\makebox(0,0)[lt]{\begin{minipage}{0.29478523\unitlength}\raggedright $t_c^+$\end{minipage}}}%
    \put(0.07997802,0.78417764){\color[rgb]{0,0,0}\makebox(0,0)[lt]{\begin{minipage}{0.50797817\unitlength}\raggedright $\frac{1}{3}$\end{minipage}}}%
    \put(0.07978079,0.6264203){\color[rgb]{0,0,0}\makebox(0,0)[lt]{\begin{minipage}{0.50797817\unitlength}\raggedright $\frac{1}{4}$\end{minipage}}}%
    \put(0.07778079,0.447936){\color[rgb]{0,0,0}\makebox(0,0)[lt]{\begin{minipage}{0.50797817\unitlength}\raggedright $t$\end{minipage}}}%
    \put(0.70814745,0.08942955){\color[rgb]{0,0,0}\makebox(0,0)[lt]{\begin{minipage}{0.50797817\unitlength}\raggedright $\frac{1}{9}$\end{minipage}}}%
  \end{picture}%
\endgroup%

%% file: paths_of_descent_and_ascent_t=031_a=011.pdf_tex
\begingroup%
  \makeatletter%
  \providecommand\color[2][]{%
    \errmessage{(Inkscape) Color is used for the text in Inkscape, but the package 'color.sty' is not loaded}%
    \renewcommand\color[2][]{}%
  }%
  \providecommand\transparent[1]{%
    \errmessage{(Inkscape) Transparency is used (non-zero) for the text in Inkscape, but the package 'transparent.sty' is not loaded}%
    \renewcommand\transparent[1]{}%
  }%
  \providecommand\rotatebox[2]{#2}%
  \ifx\svgwidth\undefined%
    \setlength{\unitlength}{309.625bp}%
    \ifx\svgscale\undefined%
      \relax%
    \else%
      \setlength{\unitlength}{\unitlength * \real{\svgscale}}%
    \fi%
  \else%
    \setlength{\unitlength}{\svgwidth}%
  \fi%
  \global\let\svgwidth\undefined%
  \global\let\svgscale\undefined%
  \makeatother%
  \begin{picture}(1,1.00121114)%
    \put(0,0){\includegraphics[width=\unitlength]{paths_of_descent_and_ascent_t=031_a=011.pdf}}%
        \put(0.52,0.99){\color[rgb]{0,0,0}\makebox(0,0)[lt]{\begin{minipage}{0.33116057\unitlength}\raggedright $\scriptstyle\op{Im}(z)$\end{minipage}}}%
    \put(0.8692436,0.5628077){\color[rgb]{0,0,0}\makebox(0,0)[lt]{\begin{minipage}{0.33116057\unitlength}\raggedright $\scriptstyle\op{Re}(z)$\end{minipage}}}%

  \end{picture}%
\endgroup%

%% file: paths_of_descent_and_ascent_t=03317_a=011.pdf_tex
\begingroup%
  \makeatletter%
  \providecommand\color[2][]{%
    \errmessage{(Inkscape) Color is used for the text in Inkscape, but the package 'color.sty' is not loaded}%
    \renewcommand\color[2][]{}%
  }%
  \providecommand\transparent[1]{%
    \errmessage{(Inkscape) Transparency is used (non-zero) for the text in Inkscape, but the package 'transparent.sty' is not loaded}%
    \renewcommand\transparent[1]{}%
  }%
  \providecommand\rotatebox[2]{#2}%
  \ifx\svgwidth\undefined%
    \setlength{\unitlength}{312.475bp}%
    \ifx\svgscale\undefined%
      \relax%
    \else%
      \setlength{\unitlength}{\unitlength * \real{\svgscale}}%
    \fi%
  \else%
    \setlength{\unitlength}{\svgwidth}%
  \fi%
  \global\let\svgwidth\undefined%
  \global\let\svgscale\undefined%
  \makeatother%
  \begin{picture}(1,0.99077926)%
    \put(0,0){\includegraphics[width=\unitlength]{paths_of_descent_and_ascent_t=03317_a=011.pdf}}%
        \put(0.52,0.99){\color[rgb]{0,0,0}\makebox(0,0)[lt]{\begin{minipage}{0.33116057\unitlength}\raggedright $\scriptstyle\op{Im}(z)$\end{minipage}}}%
    \put(0.8692436,0.5628077){\color[rgb]{0,0,0}\makebox(0,0)[lt]{\begin{minipage}{0.33116057\unitlength}\raggedright $\scriptstyle\op{Re}(z)$\end{minipage}}}%
  \end{picture}%
\endgroup%

%% file: paths_of_descent_and_ascent_t=1_over_3_a=1_over_9.pdf_tex
\begingroup%
  \makeatletter%
  \providecommand\color[2][]{%
    \errmessage{(Inkscape) Color is used for the text in Inkscape, but the package 'color.sty' is not loaded}%
    \renewcommand\color[2][]{}%
  }%
  \providecommand\transparent[1]{%
    \errmessage{(Inkscape) Transparency is used (non-zero) for the text in Inkscape, but the package 'transparent.sty' is not loaded}%
    \renewcommand\transparent[1]{}%
  }%
  \providecommand\rotatebox[2]{#2}%
  \ifx\svgwidth\undefined%
    \setlength{\unitlength}{400.69916875bp}%
    \ifx\svgscale\undefined%
      \relax%
    \else%
      \setlength{\unitlength}{\unitlength * \real{\svgscale}}%
    \fi%
  \else%
    \setlength{\unitlength}{\svgwidth}%
  \fi%
  \global\let\svgwidth\undefined%
  \global\let\svgscale\undefined%
  \makeatother%
  \begin{picture}(1,0.77356974)%
    \put(0,0){\includegraphics[width=\unitlength]{paths_of_descent_and_ascent_t=1_over_3_a=1_over_9.pdf}}%
        \put(0.52,0.99){\color[rgb]{0,0,0}\makebox(0,0)[lt]{\begin{minipage}{0.33116057\unitlength}\raggedright $\scriptstyle\op{Im}(z)$\end{minipage}}}%
    \put(0.8692436,0.5628077){\color[rgb]{0,0,0}\makebox(0,0)[lt]{\begin{minipage}{0.33116057\unitlength}\raggedright $\scriptstyle\op{Re}(z)$\end{minipage}}}%
  \end{picture}%
\endgroup%

%% file: paths_of_descent_and_ascent_t=034_a=011.pdf_tex
\begingroup%
  \makeatletter%
  \providecommand\color[2][]{%
    \errmessage{(Inkscape) Color is used for the text in Inkscape, but the package 'color.sty' is not loaded}%
    \renewcommand\color[2][]{}%
  }%
  \providecommand\transparent[1]{%
    \errmessage{(Inkscape) Transparency is used (non-zero) for the text in Inkscape, but the package 'transparent.sty' is not loaded}%
    \renewcommand\transparent[1]{}%
  }%
  \providecommand\rotatebox[2]{#2}%
  \ifx\svgwidth\undefined%
    \setlength{\unitlength}{309.59375bp}%
    \ifx\svgscale\undefined%
      \relax%
    \else%
      \setlength{\unitlength}{\unitlength * \real{\svgscale}}%
    \fi%
  \else%
    \setlength{\unitlength}{\svgwidth}%
  \fi%
  \global\let\svgwidth\undefined%
  \global\let\svgscale\undefined%
  \makeatother%
  \begin{picture}(1,1)%
    \put(0,0){\includegraphics[width=\unitlength]{paths_of_descent_and_ascent_t=034_a=011.pdf}}%
        \put(0.52,0.99){\color[rgb]{0,0,0}\makebox(0,0)[lt]{\begin{minipage}{0.33116057\unitlength}\raggedright $\scriptstyle\op{Im}(z)$\end{minipage}}}%
    \put(0.8692436,0.5628077){\color[rgb]{0,0,0}\makebox(0,0)[lt]{\begin{minipage}{0.33116057\unitlength}\raggedright $\scriptstyle\op{Re}(z)$\end{minipage}}}%
  \end{picture}%
\endgroup%

%% file: DDP_jsp.bbl
\begin{thebibliography}{10}

\bibitem{Alberts07}
B.~Alberts, ``Intracellular vesicular traffic,'' in {\em Molecular biology of
  the cell}, ch.~13, Garland Science, Taylor and Francis Group, 2007.

\bibitem{Leibler89}
S.~Leibler, R.~R.~P. Singh, and M.~E. Fisher, ``Thermodynamic behavior of
  two-dimensional vesicles,'' {\em Phys. Rev. Lett.}, vol.~59, pp.~1989--1992,
  Nov 1987.

\bibitem{Fisher91}
M.~E. Fisher, A.~J. Guttmann, and S.~G. Whittington, ``Two-dimensional lattice
  vesicles and polygons,'' {\em J. Phys. A: Math. Gen.}, vol.~24,
  pp.~3095--3106, 1991.

\bibitem{Flory53}
P.~J. Flory, ``Configuration of polymer chains,'' in {\em {Principles of
  Polymer Chemistry}}, ch.~10, Cornell University Press, 1953.

\bibitem{Richard01}
C.~Richard, A.~J. Guttmann, and I.~Jensen, ``{Scaling function and universal
  amplitude combinations for self-avoiding polygons},'' {\em J. Phys. A: Math.
  Gen.}, vol.~34, pp.~L495--501, 2001.

\bibitem{NIST}
``{\it NIST Digital Library of Mathematical Functions}.''
  http://dlmf.nist.gov/, Release 1.0.13 of 2016-09-16.
\newblock F.~W.~J. Olver, A.~B. {Olde Daalhuis}, D.~W. Lozier, B.~I. Schneider,
  R.~F. Boisvert, C.~W. Clark, B.~R. Miller and B.~V. Saunders, eds.

\bibitem{Prellberg95}
T.~Prellberg, ``Uniform $q$-series asymptotics for staircase polygons,'' {\em
  J. Phys. A: Math. Gen.}, vol.~28, pp.~1289--1304, 1995.

\bibitem{Haug15}
N.~Haug and T.~Prellberg, ``{Uniform asymptotics of area-weighted Dyck
  paths},'' {\em J. Math. Phys.}, vol.~56, p.~043301, 2015.

\bibitem{Haug16}
N.~Haug, T.~Prellberg, and G.~Siudem, ``Area-width scaling in generalised
  {M}otzkin paths.'' \url{https://arxiv.org/pdf/1605.09643v2.pdf}.

\bibitem{Mitra10}
M.~K. Mitra, G.~I. Menon, and R.~Rajesh, ``Asymptotic behaviour of convex and
  column-convex lattice polygons with fixed area and varying perimeter,'' {\em
  Journal of Statistical Mechanics: Theory and Experiment}, vol.~2010, no.~07,
  p.~P07029, 2010.

\bibitem{Temperley52}
H.~N.~V. Temperley, ``{Statistical mechanics and the partition of numbers II.
  The form of crystal surfaces },'' {\em Math. Proc. Cambridge Philos. Soc.},
  vol.~48, pp.~683--697, 1952.

\bibitem{Dietrich88}
S.~Dietrich, ``Wetting phenomena,'' in {\em Phase Transitions and Critical
  Phenomena} (C.~Domb and J.~L. Lebowitz, eds.), vol.~12, ch.~1, pp.~2--218,
  Academic Press, 1988.

\bibitem{Cardy01}
J.~Cardy, ``{Exact scaling functions for self-avoiding loops and branched
  polymers},'' {\em Journal of Physics A: Mathematical and General}, vol.~34,
  no.~47, p.~L665, 2001.

\bibitem{Viennot89}
{X. G. Viennot}, ``{Heaps of pieces, I: Basic Definitions and Combinatorial
  Lemmas},'' {\em Annals of the New York Academy of Sciences}, vol.~576,
  pp.~542--570, 1989.

\bibitem{Carlitz75}
L.~Carlitz, ``{Fibonacci Notes 4: $q$-Fibonacci Polynomials},'' {\em The
  Fibonacci Quarterly}, vol.~13(2), pp.~97--102, 1975.

\bibitem{Bousquet02}
M.~Bousquet-M{\'e}lou and A.~Rechnitzer, ``Lattice animals and heaps of
  dimers,'' {\em Discrete Mathematics}, vol.~258, pp.~235--274, 2002.

\bibitem{Flajolet80}
P.~Flajolet, ``{Combinatorial Aspects of Continued Fractions},'' {\em Discrete
  Mathematics}, vol.~32, pp.~125--161, 1980.

\bibitem{Gasper90}
G.~Gasper and M.~Rahman, {\em Basic Hypergeometric Series}, vol.~96 of {\em
  Encyclopedia of Mathematics and its Applications}.
\newblock Cambridge University Press, 1990.

\bibitem{Rensburg00}
E.~J.~J. van Rensburg, {\em {The Statistical Mechanics of Interacting Walks,
  Polygons, Animals and Vesicles}}, vol.~18 of {\em Oxford Lecture Series in
  Mathematics and its Applications}.
\newblock Oxford University Press, 2000.

\bibitem{Flajolet09_8}
P.~Flajolet and R.~Sedgewick, ``Saddle point asymptotics,'' in {\em {Analytic
  Combinatorics}}, ch.~8, Cambridge University Press, 2009.

\bibitem{Ursell72}
F.~Ursell, ``{Integrals with a large parameter. Several nearly coincident
  saddle points},'' {\em {Mathematical Proceedings of the Cambridge
  Philosophical Society}}, vol.~72, pp.~49--65, 7 1972.

\bibitem{Paris91}
R.~B. Paris, ``{The Asymptotic Behaviour of Pearcey's Integral for Complex
  Variables},'' {\em Proceedings: Mathematical and Physical Sciences},
  pp.~391--426, 1991.

\end{thebibliography}
